\documentclass[pra,aps,twocolumn,superscriptaddress,longbibliography]{revtex4-1}
\usepackage[colorlinks=true, citecolor=red, urlcolor=blue ]{hyperref}
\usepackage{graphicx}
\usepackage{bm}
\usepackage{amsmath,amsfonts}
\usepackage{amsthm}
\usepackage{hyperref}
\usepackage{xcolor}
\hypersetup{
    urlcolor=magenta,
    citecolor=blue
           }

\usepackage{braket}
\theoremstyle{plain}
\usepackage{color}
\usepackage{amssymb}
\usepackage{amsthm}
\usepackage{amsfonts}
\usepackage{float}
\usepackage{tabularx}
\usepackage{graphicx}

\usepackage{mathtools}
\usepackage{esvect}
\usepackage{wrapfig}
\usepackage{amsthm}
\usepackage{verbatim}
\usepackage{bbm}
\usepackage[normalem]{ulem}

\usepackage{enumitem}
\usepackage{fmtcount}
\usepackage{booktabs}
\usepackage{csquotes}
\usepackage{epsfig}

\usepackage{tabularx}
\usepackage{graphicx}
\usepackage{amsmath}
\usepackage{braket}
\usepackage{latexsym}
\usepackage{bm}
\usepackage{graphics,epstopdf}
\usepackage{enumitem}
\usepackage{fmtcount}
\usepackage{booktabs}
\usepackage{csquotes}
\usepackage{epsfig}

\theoremstyle{plain}

\def\bea{\begin{eqnarray}}
\def\eea{\end{eqnarray}}
\def\ba{\begin{array}}
\def\ea{\end{array}}

\def\beq{\begin{equation}}
\def\eeq{\end{equation}}

\usepackage[normalem]{ulem}
\usepackage{float}
\usepackage{graphicx}  
\usepackage{dcolumn}          
\usepackage{amssymb}
\usepackage{appendix}
\usepackage{physics}   
\usepackage{mathtools}
\usepackage{esvect}
\usepackage{wrapfig}
\usepackage{amsthm}
\usepackage{verbatim}
\usepackage{bbm}

\usepackage[mathscr]{euscript}
\def\Tr{\operatorname{Tr}}

\def\({\left(}
\def\){\right)}
\def\[{\left[}
\def\]{\right]}

\newtheorem{example}{Example}

\newtheorem{lemma}{Theorem}

\begin{document}
\title{Designing a Universal Quantum Switch for Arbitrary Quantum Dynamics}

\author{Priya Ghosh}
\affiliation{Harish-Chandra Research Institute,  A CI of Homi Bhabha National
Institute, Chhatnag Road, Jhunsi, Prayagraj - 211019, India}
\author{Kornikar Sen}
\affiliation{Departamento de Física Teórica, Universidad Complutense, 28040 Madrid, Spain}
\author{Ujjwal Sen}
\affiliation{Harish-Chandra Research Institute,  A CI of Homi Bhabha National
Institute, Chhatnag Road, Jhunsi, Prayagraj - 211019, India}

\begin{abstract}
{A quantum switch is a superoperator that, in general, creates a superposition of various causal orders of two or more  quantum dynamics that are all divisible in the complete positivity (CP) sense. We introduce a process that we term as the universal quantum switch (UQS), which unlike conventional quantum switches, allows for the construction of a quantum switch that can superpose different causal orders of any set of quantum dynamics, regardless of their CP-divisibility.} {Our approach also enables the construction of a quantum switch while considering a single environment connected with the system, in contrast to the traditional one.}
Moreover, we show the UQS provides more advantages in performance for a certain state discrimination task compared to traditional quantum switches. The next question that we address is the following: What is the CP-divisibility characteristic of a dynamics built by acting a quantum switch on CP-divisible or -indivisible dynamics? In this regard, an example is presented where the dynamics created by the action of the UQS on two CP-indivisible dynamics is CP-indivisible. Additionally, we prove a necessary and sufficient condition for the channel created by acting  the traditional quantum switch on two CP-divisible dynamics to be CP-divisible. Furthermore, we present some examples of CP-divisible dynamics on which, when the usual quantum switch is operated, the resulting dynamics not only becomes CP-indivisible but also turns into P-indivisible. Our findings demonstrate that quantum switches can build CP-divisible, CP-indivisible, and even P-indivisible dynamics from CP-divisible dynamics, underscoring the versatility of this technique. 
\end{abstract}

\maketitle

\section{Introduction}
\label{sec1} 
Quantum evolutions of 
isolated systems  can be described using unitary maps. However, when a system interacts with an external entity, commonly referred to as the environment, the evolution of the system over time cannot be expressed solely through a unitary transformation. In such scenarios, the system is considered to be open to the environment \cite{Rivas-open,RHP-review,BLPV-review, VA-review}.
If the system and the environment are initially in a product state, the evolution of the system from its initial state to any subsequent state can be elucidated using a linear completely positive trace-preserving map (CPTP). {Interestingly, the converse is also true, i.e., any CPTP evolution of a system can always be explained by considering the presence of an environment and the composite dynamics of the system-environment set-up to be unitary.}

CPTP dynamics of systems can mainly be classified in two ways: (i) CP-divisible and CP-indivisible, and (ii) P-divisible and P-indivisible. 
CP- and P-indivisibility find extensive applications in a large variety of domains, ranging from quantum information theory \cite{qkd-nm,nonlocality-nm,rivu-1,majumdar,rivu-2}, quantum computation \cite{chen}, to quantum thermodynamics~\cite{thermo-nm-100,ahana-refrigerator,srijon-battery-nm,ahana-heat-current,nm-thermo-95,kornikar-2nd-battery,tang}.
These distinct classes of dynamics, along with their respective quantifiers have been studied in detail in recent years~\cite{RHP-measure,BLP-measure}. However, a direct relation between the quantifiers of CP- and P-divisibility is not yet known. 

An intriguing phenomenon offered by the striking characteristics of quantum mechanics is the chance to superpose different causal orders of two or more dynamics, which can result in an indefinite causal order~\cite{ICO-1,ICO-2,ICO-3,ICO-4,ICO-5}. A popular method for the implementation of indefinite causal orders is known as the quantum switch~\cite{switch-first-paper,tang}. A quantum switch is basically a supermap that takes a set of dynamics and produces a new dynamics that represents the indefinite causal order of the given set of dynamics.
The application of indefinite causal order of channels on systems has been found to provide an advantage in various quantum mechanical tasks, viz. channel discrimination~\cite{channel-discrimination-1,switch-channel-discrimination-2}, nonlocal games~\cite{ICO-1}, quantum communication~\cite{communication-switch-3,communication-switch-1,pati-switch,communication-switch-2,communication-ICO-51}, metrology~\cite{metrology-ICO-91,metrology-switch-1,metrology-switch-2,metrology-switch-3,metrology-switch-4,QFI-ICO-102,QFI-ICO-103}, open system dynamics~\cite{open-switch-1,open-switch-2,open-switch-3,open-switch-4,samyadeb-switch}, state discrimination~\cite{switch-task-1,switch-state-discrimination}, etc. Quantum switches have been produced in various experiments using photonic set-ups~\cite{expt-switch-1,expt-switch-2,expt-switch-3,expt-switch-4,expt-switch-5}.

{To study the action of a quantum switch, one needs to know the Kraus decomposition of each of the channels fed into the switch for all time intervals, even at the intermediate times of the evolution of the system on which the maps act~\cite{switch-first-paper}. {Thus, all the dynamics that are being fed into the quantum switch must be CP-divisible.}
In our work, we introduce a quantum switch model that also constructs the indefinite causal order of two or more quantum dynamics. But in this case, to build the model, it's not necessary to have the Kraus decomposition of the maps, which implies that the introduced quantum switch can act on any set of CPTP dynamics, even on NCPTP dynamics.
Therefore, we name it the ``universal quantum switch" (UQS). We refer to the previous model of the quantum switch as the conventional quantum switch (CQS). The UQS creates the superposition of different causal orders of dynamics in a distinct way. Considering a state discrimination task, we demonstrate that the UQS can provide an advantage over the CQS.

The next goal of this article is to investigate the CP-divisible characteristic of the indefinite causal order of dynamics. In this regard, we first prepare an indefinite causal order of a pair of CP-indivisible dynamics and show that the produced dynamics is also CP-indivisible. 
{Since the CQS can not act on CP-indivisible dynamics, in this case,} we used the UQS to build the indefinite causal order of the CP-indivisible dynamics.

Furthermore, we address the following question:
Does CP-divisibility remain preserved under the action of a quantum switch? In particular, if we compose a dynamics by applying a quantum switch to two arbitrary CP-divisible dynamics, will the resulting dynamics also be CP-divisible? By using the CQS, we find that the necessary and sufficient condition for the final dynamics to be CP-divisible is that each of the Kraus operators of one of the initial CP-divisible dynamics should commute with the same of the other dynamics.
Finally, by exploring a few examples, we show the dynamics, constructed by applying the CQS to a pair of CP-divisible dynamics, can even be P-indivisible.
These results reveal that the quantum switch acting upon CP-divisible dynamics is capable of creating CP-divisible, CP-indivisible, and even P-indivisible dynamics, demonstrating the applicability of this approach.

The rest of the paper is organized as follows: In Sec. \ref{sec2}, we briefly recapitulate CPTP, CP-divisible, CP-indivisible, P-divisible, and P-indivisible quantum dynamics. We also include a discussion on the detection of P-indivisibility, conventional quantum switches, and different typical CP-divisible channels, e.g., phase damping, depolarizing, and amplitude damping noise, in the same section. In Sec. \ref{sec3}, we introduce the universal quantum switch, which can act on any two or more quantum dynamics even if the dynamics are CP-indivisible. A state discrimination task in which the universal quantum switch can outperform the typical quantum switch in a range of instances is discussed in Sec. \ref{sdt}. In Sec. \ref{sec4}, we discuss an example of a pair of CP-indivisible dynamics on which, when the universal quantum switch acts, the resulting dynamics remains CP-indivisible.
We derive a necessary and sufficient criterion for the action of the conventional quantum switch on two arbitrary CP-divisible dynamics to produce a CP-divisible dynamics in Sec. \ref{sec5}. Examples of pairs of CP-divisible dynamics that lose their CP-divisibility when applied in indefinite causal order modeled by the CQS are also presented in the same section, i.e., Sec. \ref{sec5}. Finally, we present our concluding remarks in Sec. \ref{sec6}.

\section{Preliminaries}
\label{sec2}
While the postulates of quantum mechanics were initially developed for isolated quantum systems, in reality, systems often involve unavoidable interactions with their surroundings.  In such cases, the system is considered open and susceptible to external influences. Consequently, the system's dynamics typically deviate from unitarity, even though the combined dynamics of the system and its environment, with which it interacts, remain unitary. In this section, we will briefly recapitulate the evolution of these open systems. 
\vspace{1mm}\\
\textbf{Completely-positive trace-preserving maps:} 
Let us denote the Hilbert spaces describing the states of the quantum system and environment by $\mathcal{H}_S$ and $\mathcal{H}_E$, respectively. Moreover, let the state of the system at time $t=t_1$ and the map describing the system's evolution between, say time $t=t_1$ to $t=t_2\geq t_1$, be $\rho_S (t_1)$ and $\Lambda(t_2,t_1)$, respectively. Hence we can write $\rho_S (t_2)=\Lambda(t_2,t_1)\rho_S (t_1)$.  It is well known that if initially, i.e., at $t=0$, the environment and the system are in a product state, then the evolution of the system with time due to its interaction with the environment can be tracked using completely positive trace-preserving (CPTP) maps~\cite{chuang}. Since any map is CPTP, if and only if, it has Kraus operator decomposition~\cite{chuang}, $\rho_S (t)$ can be expressed in terms of the Kraus operator decomposition of the map, $\Lambda(t,0)$, acting on $\rho_S(0)$, i.e.,
$\rho_S (t) =\Lambda(t,0) \rho_S(0)=\sum_i K_i(t,0) \rho_S(0) K_i^\dagger (t,0)$. Here $\{K_i (t,0) \}_i$ represents the set of Kraus operators of the map $\Lambda (t,0)$ and it
satisfies $\sum_i K_i^\dagger (t,0) K_i (t,0) = \mathbbm{I}$, where $\mathbbm{I}$ denotes the identity operator acting on $\mathcal{H}_S$. From now on, we will denote the set of maps, $\{\Lambda(t,0)\}_t$, that describe the dynamics of the system from an initial time as $\Lambda$. We will always consider the maps, $\Lambda(t,0)$, to be CPTP. However, the maps $\Lambda(t_2,t_1)$, for $t_1\neq 0$,
 which represents the evolution of the system in an intermediate time may not be CPTP. If the map, $\Lambda(t,0)$, is CPTP for all $t$, we call the dynamics, $\Lambda$, as CPTP dynamics or CPTP evolution.

The CPTP dynamics, $\Lambda$, of systems can be classified in the following two ways ~\cite{RHP-review,BLPV-review,VA-review,Rivas-open}:
\begin{enumerate}
    \item  {CP-divisible and CP-indivisible quantum processes}: Any quantum dynamics, say $\Lambda$, describing the time evolution of the state, $\rho_S$, of the system is said to be CP-divisible if all the elements of $\Lambda$ satisfy the following relation:
\begin{align}
\label{eq-divisible}
\Lambda(t_2, 0) = \Lambda(t_2, t_1) \cdot \Lambda(t_1, 0)
\end{align}
for all $t_2 \geq t_1 \geq 0$, where the map, $\Lambda(t_2,t_1)$, is CPTP. 
Otherwise, it is referred to as a CP-indivisible quantum dynamics. To quantify the CP-divisibility of quantum dynamics, the Rivas-Huelga-Plenio measure can be utilized~\cite{RHP-measure}.
\item
 P-divisible and P-indivisible quantum operations~\cite{RHP-review,BLPV-review,VA-review,Rivas-open}: A quantum map, $\Lambda(t_2,0)$, is called P-divisible if the dynamics satisfies Eq.~\eqref{eq-divisible}, where $\Lambda(t_2,t_1)$ is positive and trace-preserving for all $t_2\geq t_1\geq 0$. If a dynamics is not P-divisible, it is called a P-indivisible quantum dynamics. P-indivisibility of dynamics can be measured using the Breuer-Laine-Pillo (BLP) measure~\cite{BLP-measure}. 
\end{enumerate}

Some typical examples of qubit CP-divisible channels are discussed below:
\begin{itemize}
    \item Ideal phase damping channel: The ideal phase damping channels are classic examples of CP-divisible dynamics. An ideal phase damping channel, $\Lambda_{{PDC}}(t,0)$, can be expressed using the following Kraus operators: $K_0(t,0)=e^{-\Gamma_{PDC}t/2} \mathbbm{I}_2$, $K_1(t,0)=\sqrt{1-e^{-\Gamma_{PDC}t}}\sigma_z$, where $ \Gamma_{PDC}$, $\mathbb{I}_2$, and $\sigma_z$ represent the Lindblad coefficient of the channel, identity operator acting on qubit Hilbert space, and the Pauli-$z$ matrix, respectively. The action of the channel, $\Lambda_{{PDC}}(t,0)$, on any single-qubit state, $\rho_S$, can be expressed as 
\begin{equation*}
 \Lambda_{{PDC}} (t,0)\rho_S = e^{-\Gamma_{PDC}t} \rho_S + (1-e^{-\Gamma_{PDC}t}) \sigma_z \rho_S \sigma_z.
\end{equation*}
\item  Ideal depolarizing channel: {The ideal depolarizing maps increase the mixedness of all quantum states in the same ratio}. It is another example of a CP-divisible channel.
The transformation of any state, $\rho_S$, under the action of an ideal depolarizing channel, $\Lambda_{{DC}} (t,0)$, having Lindblad coefficient, $\Gamma_{DC}$, is given by $\Lambda_{{DC}}(t,0) \rho_S = {\left(1+3e^{-\Gamma_{DC}t}\right)} \rho/{4} +\left(1-e^{-\Gamma_{DC}t}\right) (\sigma_x \rho \sigma_x + \sigma_y \rho \sigma_y + \sigma_z \rho \sigma_z)/4,$ where $\sigma_x$, $\sigma_y$, and $\sigma_z$ are the Pauli matrices. The Kraus operators of this noise are 
$K_0(t,0) = \frac{\sqrt{\left(1+3e^{-\Gamma_{DC}t}\right)}}{2} \mathbbm{I}_2$, $K_1(t,0) = \frac{\sqrt{\left(1-e^{-\Gamma_{DC}t}\right)}}{2} \sigma_x$, $K_1(t,0) = \frac{\sqrt{\left(1-e^{-\Gamma_{DC}t}\right)}}{2} \sigma_y$, and $K_1(t,0) = \frac{\sqrt{\left(1-e^{-\Gamma_{DC}t}\right)}}{2} \sigma_z$

\item Ideal amplitude damping channel: The last example of CP-divisible dynamics that we present here is the ideal amplitude damping channel. This channel models the decay of an excited state of a qubit as a result of the spontaneous emission of a photon.
The Kraus operators of the ideal amplitude damping channel ($\Lambda_{{ADC}} (t,0)$) are given by
\begin{eqnarray}
K_0(t,0) &=& \left( \begin{array}{cc}
1 & 0 \\
0 & e^{-\Gamma_{ADC}t/2} \end{array} \right),\label{extra1} \\
K_1(t,0) &=& \left( \begin{array}{cc}
0 & \sqrt{1-e^{-\Gamma_{ADC}t}} \\
0 & 0 \end{array} \right),   \label{adc}  
\end{eqnarray} where $\Gamma_{ADC}$ denotes the Lindblad coefficient of the channel.
\end{itemize}
The parameters, $1-e^{-\Gamma_{PDC}t}$, $1-e^{-\Gamma_{DC}t}$, and $1-e^{-\Gamma_{ADC}t}$ are also known as the noise strengths of phase damping, depolarizing, and amplitude damping channels, and can be denoted as $p_{PDC}$, $p_{DC}$, and $p_{ADC}$, respectively.
The noise strengths of all of the three above-mentioned channels vary from 0 to 1.

One can notice simply from the definitions of P- and CP-divisible quantum dynamics that all CP-divisible quantum dynamics will also be P-divisible; however, the reverse is not true. Therefore, we can split up all of the quantum dynamics into two groups: P-divisible and P-indivisible. Among the set of P-divisible dynamics, we can find dynamics that are also CP-divisible. To provide a visual idea of these categories, in Fig. \ref{fig1}, we draw a schematic diagram describing the different classes of quantum dynamics discussed above. In the figure, the red solid zigzag line differentiates P-divisible dynamics from P-indivisible ones, and the black portion shows CP-divisible dynamics. The region outside of the black portion represents CP-indivisible dynamics. In the next part of the paper, we offer an example where the dynamics found by operating the UQS on two CP-indivisible dynamics is CP-indivisible. Moreover, we state a necessary and sufficient condition for the indefinite causal order of two CP-divisible dynamics, modeled by the CQS, to be CP-divisible. Also, a few examples of pairs of CP-divisible dynamics are presented that become P-indivisible when acted on in indefinite causal order, again modeled by the CQS. These results are symbolically depicted using the purple arrows in Fig. \ref{fig1}. There can be other examples also for which, say the action of the switch may map P-divisible maps to  

Before going into the details of the results, let us first discuss the tools we need to arrive at them.

\begin{figure}[h!]
\includegraphics[scale=0.5]{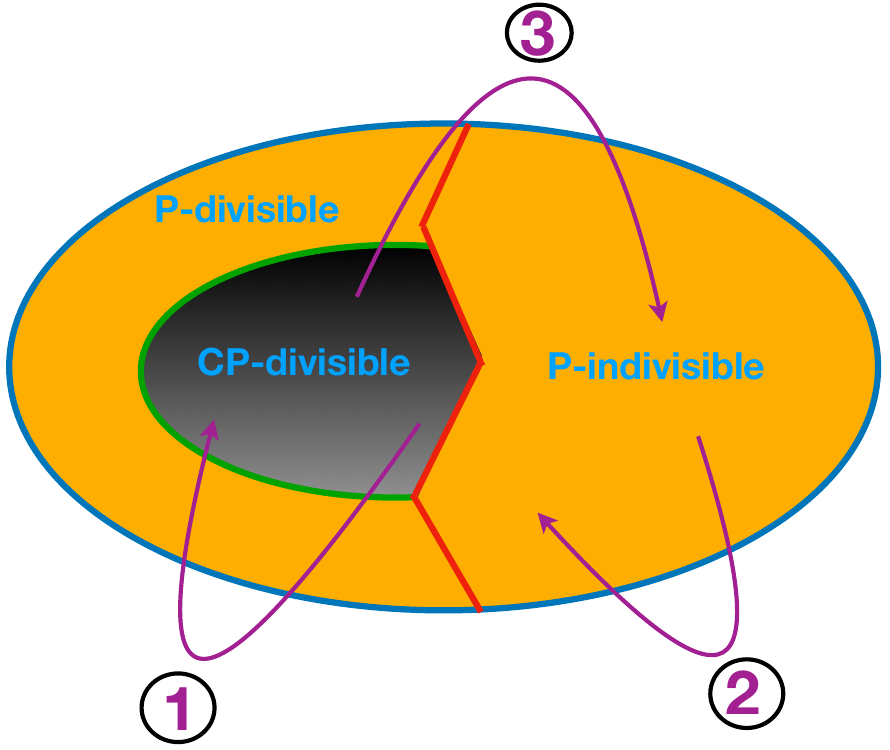}
\caption{\textbf{Schematic representation of different categories of quantum dynamics.} The set of all completely positive trace-preserving quantum dynamics is convex and is shown by a blue-colored spheroid filled with yellow. The left and right sides of the spheroid, divided using the red zigzag curve, indicate the sets of P-divisible and P-indivisible dynamics, respectively. In contrast, the green-colored curve helps to encompass the set of all CP-divisible quantum dynamics and is marked by a black-shaded region. The unshaped yellow region, therefore, represents the set of all CP-indivisible quantum dynamics. We have found examples where the action of the quantum switch maps 1) a pair of CP-divisible maps to CP-divisible maps, 2) CP-divisible maps to P-indivisible maps, and 3) P-indivisible maps to itself. These three examples are depicted in the diagram using the violet arrows and corresponding numbers. }
\label{fig1} 
\end{figure}

We begin by providing a discussion on the detection of the P-indivisibility of quantum dynamics. \vspace{1mm}\\
\textbf{Detection of P-indivisibility~\cite{BLP-measure}:} When two different states of a system are affected by any fixed P-divisible dynamics, the distinguishability between the states either decreases or remains the same but never increases with time. Utilizing this characteristic of P-divisible dynamics, the BLP measure was introduced~\cite{BLP-measure}, which can quantify the P-indivisibility of quantum dynamics. 
In this paper, we will not exactly measure the P-indivisibility of any dynamics. Our aim will always be to detect P-indivisibility of quantum dynamics. As we just mentioned, if a quantum dynamics, say $\Lambda^P $, of a particular system is P-divisible, then the distinguishability between the states, $\Lambda^P(t,0)\rho^1_S$ and $\Lambda^P(t,0) \rho^2_S$, will monotonically decrease with time, $t$, for any two initial states, $\rho^1_S$ and $\rho^2_S$, of the system. Hence the trace distance, $D(\Lambda^P(t,0)\rho^1_S,\Lambda^P(t,0) \rho^2_S)$, between the states, $\Lambda^P(t,0)\rho^1_S$ and $\Lambda^P(t,0)\rho^2_S$, 
 will never increase with time, $t$. In other words, P-divisible quantum dynamics, $\Lambda^P$, imply
\begin{align}
D(\Lambda^P(t_2,0) \rho^1_S,\Lambda^P(t_2,0) \rho^2_S)\nonumber\\ \leq D(\Lambda^P(t_1,0) \rho^1_S,\Lambda^P(t_1,0) \rho^2_S), \label{eq-BLP}
\end{align}
for any $t_2\geq t_1 \geq 0$ and any pair of states, $\rho_S^1$ and $\rho_S^2$.
Hence if a quantum dynamics violate the above inequality for any pair of times, $\{t_1,t_2\}$, and states, $\{\rho^1_S,\rho^2_S\}$, then the quantum dynamics is certainly P-indivisible.

For completeness, let us briefly define the trace distance.
The trace distance between two states, $\rho^1_S$ and $\rho^2_S$, acting on the same Hilbert space, can be expressed as
\begin{equation*}
D(\rho^1_S, \rho^2_S) \coloneqq \frac{1}{2} ||\rho^1_S - \rho^2_S||,
\end{equation*}
where $||A||$, for any matrix $A$, denotes $\tr(\sqrt{A^\dagger A})$.

In the next part, which is the last part of this section, we will discuss the conventional quantum switches.\vspace{1mm}\\
\textbf{Conventional Quantum Switches \cite{switch-first-paper}}: Quantum switches are supermaps which take two or more quantum dynamics and create a superposition of different causal orders of them. We will denote the action of a quantum switch as $\mathcal{S}$. As the simplest non-trivial case, we consider quantum switches that are able to superpose causal orders of two quantum dynamics, say $\Lambda_1$ and $\Lambda_2$. Let both of these channels act on a system which is initially prepared in the state $\rho_S$. Typically a control system is required in order to build a quantum switch. Since we focus on the superposition of causal orders of two quantum dynamics, the corresponding control system needs to have at least dimension two. 
Let us assume that when the control is in the state $\ket{0}$, the quantum map, $\Lambda_1(t_1,0)$, acts on the state of the system, {followed by} the operation of the quantum map $\Lambda_2(t_2,t_1)$, and the channels act in the opposite order, i.e., the composite map, $\Lambda_1(t_2,t_1)\Lambda_2(t_1,0)$, acts on the state of the system when the control is in the state $\ket{1}$. Hence the state of the system first evolves for time $0$ to $t_1$ and then $t_1$ to $t_2$. These are the two possible causal orders of the pair of dynamics. 
Here $\ket{0}$ and $\ket{1}$ are elements of the two-dimensional computational basis of the Hilbert space, $\mathcal{H}_C$, which describes the control qubit. We consider $\{\ket{0},\ket{1}\}$ to be the eigenbasis of $\sigma_z$.
Let us consider the initial states of the control qubit and the system to be $\omega_C$ and $\rho_S$ respectively. Then the action of the dynamics, created by applying a quantum switch, $\mathcal{S}$, on the two quantum dynamics, $\Lambda_1$ and $\Lambda_2$, will transform the state, $\rho_S$, as 
\begin{eqnarray}
&&\rho'_S= {}_C \bra{\pm}\mathcal{S}(\Lambda_1, \Lambda_2) (\rho_S \otimes \omega_C)\ket{\pm}_C/N,\nonumber \\ &&\text{where }
\mathcal{S}(\Lambda_1, \Lambda_2) (\rho_S \otimes \omega_C) \coloneqq \sum_{ij} W_{ij} (\rho_S \otimes \omega_C) W_{ij}^{\dagger},\nonumber\\
&&\text{and } W_{ij} \coloneqq K_i^{(1)}(t_2,t_1) K_j^{(2)}(t_1,0) \otimes \ket{0} \bra{0} \nonumber\\&&+ K_j^{(2)}(t_2,t_1) K_i^{(1)}(t_1,0) \otimes \ket{1} \bra{1}, \nonumber\\ \label{eq1}
\end{eqnarray}
where $K_i^{(l)}(t_2,t_1)$ and $K_i^{(l)}(t_1,0)$ are Kraus operators of the maps $\Lambda_l(t_2,t_1)$ and $\Lambda_l(t_1,0)$, respectively and $N$ is the normalization coefficient. Here $l$ takes values $1$ and $2$ representing the two dynamics, $\Lambda_1$ and $\Lambda_2$, and $\ket{\pm}_C=(\ket{0}\pm\ket{1})/\sqrt{2}$ is element of $\mathcal{H}_C$. Since the map, $\mathcal{S}(\Lambda_1, \Lambda_2)$, involves the action of two consecutive maps, the total time for which the system evolves under the entire process is $0$ to $t_2$. We will denote the map which transfers $\rho_S$ to $\rho'_S$, constructed using the quantum switch, $\mathcal{S}$, as $\Sigma(\Lambda_1,\Lambda_2,t_2,0,\cdot)$. The action of the map on a state, $\rho_S$, will be denoted as $\Sigma(\Lambda_1,\Lambda_2,t_2,0,\rho_S )$. 
Though the exact form of the map also depends on the state of the control qubit, $\omega_C$, and the parameter, $t_1$, we are not explicitly including them in the notation, $\Sigma(\Lambda_1,\Lambda_2,t_2,0,\cdot)$, for the sake of simplicity. The fourth argument of the function, $\Sigma$, denotes the time of the application of the map, $\mathcal{S}(\Lambda_1,\Lambda_2)$, on the system-control state, which in this scenario is considered to be 0.

One can generalize the definition of a quantum switch acting on two quantum dynamics to define quantum switches which act on $N$ quantum dynamics by considering a control system of dimension $N!$.

We will present another model of a quantum switch in the next section, named the universal quantum switch. The motivation behind the introduction of the quantum switch will also be discussed. 

\section{Universal Quantum Switches}
\label{sec3}
Before going into the details of the construction of a universal quantum switch, let us first discuss the need to introduce a new quantum switch.

Let us consider a situation where a system, $S$, interacts with its environment, $E$, for time $0$ to $t_2$. Moreover, consider the interaction to be of two possible types, and the evolution of $SE$ due to these two types of interactions be either described by a unitary, $U_1(t_2,0)$, or $U_2(t_2,0)$. The initial states of $S$ and $E$ are such that CPTP maps, $\Lambda_1(t_2,0)$ and $\Lambda_2(t_2,0)$, can describe the evolution of $S$ resulting from the unitaries, $U_1(t_2,0)$ and $U_2(t_2,0)$, respectively. For simplicity, let us consider the initial state of $SE$ to be a product.

We first try to apply the conventional quantum switch, $\mathcal{S}$, to this pair of evolutions, $\Lambda_1$ and $\Lambda_2$.
In this regard, let us divide the whole time range, $[0,t_2]$, into two parts, i.e., $[0,t_1)$ and $[t_1,t_2]$, and introduce a control qubit. When the state of the control qubit is $\ket{0}$ ($\ket{1}$), $\Lambda_2(t_2,t_1)\Lambda_1(t_1,0)$ ($\Lambda_1(t_2,t_1)\Lambda_2(t_1,0)$) acts on $S$. The dynamics, $\Lambda_1$ and $\Lambda_2$, being any arbitrary pair of CPTP evolutions, may not, in general, be CP-divisible. Therefore, even though the dynamics, $\Lambda_1$ and $\Lambda_2$, are CPTP, the maps, $\Lambda_1(t_2,t_1)$ and $\Lambda_2(t_2,t_1)$, may not have Kraus operator decomposition.
One can notice from Eq.~\eqref{eq1} that without Kraus operators, we can not define quantum switches in the usual way. Therefore, to define the superposition of different causal orders of CPTP dynamics, one or more of which are CP-indivisible, we need a quantum switch, which is not defined in terms of Kraus operators.

To understand the reason behind the difficulty in operating the conventional quantum switch in a more detailed way, let us discuss the scenario more elaborately. When the control qubit is, say, in state $\ket{0}$, as we mentioned above, the map $\Lambda_2(t_2,t_1)\Lambda_1(t_1,0)$ acts on $S$, i.e., $SE$ first evolves under the unitary $U_1(t_1,0)$ and then under $U_2(t_2,t_1)$. Since initially $SE$ is in a product state, the first evolution of $S$ between $[0,t_1)$ is CPTP. But during this interaction, $S$ and $E$ may have generated a quantum correlation between them, due to which the evolution of $S$ in the time range $[t_1,t_2]$ will no longer be CPTP. Hence, the evolution of $S$ from $t_1$ to $t_2$ may not be expressible in terms of Kraus operators. Instead, if we had introduced another new separate environment, $E_2$, at time $t_1$, and evolved $SE_2$ with unitary $U_2(t_2,t_1)$, then the new map, $\Lambda'_2(t_2,t_1)$, describing the evolution of the system $S$ within time $[t_1,t_2]$, would have been CPTP. Therefore, in such a scenario, the CQS could have been utilized. In this article, our purpose is to not introduce any new environment, $E_2$, and create an indefinite causal order of interactions of $S$ with the same environment, $E$. Therefore, we construct the universal quantum switch, $\mathcal{U}$, that can be used to create a superposition of causal orders of any quantum channels, $\Lambda_1$ and $\Lambda_2$, including CP-indivisible ones.

Let us now move to the construction of $\mathcal{U}$. We consider the initial state of the system to be $\rho_S$ and denote the Hilbert space of the system as $\mathcal{H}_S$. For ease of understanding, we focus on systems of dimension two, but the method of applying $\mathcal{U}$ on dynamics can easily be generalized for higher-dimensional systems.
$\Lambda_1(t_2,t_1) \Lambda_2(t_1,0) \rho_S$ and $\Lambda_2(t_2,t_1) \Lambda_1(t_1,0) \rho_S$ represent the sequential action of the quantum dynamics, $\Lambda_1$ and $\Lambda_2$, on the state $\rho_S$, in the two possible causal orders. Let both $\Lambda_1 (t_2,t_1)\Lambda_2(t_1,0) \rho_S$ and $\Lambda_2(t_2,t_1) \Lambda_1(t_1,0) \rho_S$ act on the same Hilbert space, $\mathcal{H}_S$, of dimension two.
We denote the action of the UQS on $\Lambda_1$ and $\Lambda_2$ as $\mathcal{U}(\Lambda_1,\Lambda_2)$.
To apply $\mathcal{U}(\Lambda_1,\Lambda_2)$ on $\rho_S$, we need to follow the steps presented below:
\begin{enumerate}
    \item 
    Determine the spectral decomposition of each of the quantum states: $\Lambda_1(t_2,t_1) \Lambda_2(t_1,0) \rho_S = \sum_{i=1}^2 \lambda_i \ket{\lambda_i} \bra{\lambda_i}$ and $\Lambda_2(t_2,t_1) \Lambda_1(t_1,0) \rho_S = \sum_{i=1}^2 \mu_i \ket{\mu_i} \bra{\mu_i}$.  Here $\lbrace \lambda_i \rbrace_i$ ($\lbrace \mu_i \rbrace_i$) and $\lbrace \ket{\lambda_i} \rbrace_i$ ($\lbrace \ket{\mu_i} \rbrace_i$) represent, respectively, the eigenvalues and eigenstates of $\Lambda_1(t_2,t_1) \Lambda_2 (t_1,0)\rho_S$ ($\Lambda_2(t_2,t_1) \Lambda_1(t_1,0) \rho_S$).
    \item 
    Find a basis that has a maximum overlap with both the spectral bases, $\lbrace \ket{\lambda_i} \rbrace_i$  and $\lbrace \ket{\mu_i} \rbrace_i$. 
This can be realized by finding the maxima of the following function:
\begin{eqnarray}
F(\rho_S,\Lambda_1,\Lambda_2)\coloneqq \sum_{i=1}^2 [ |\bra{\chi}\lambda_i \rangle| + |\bra{\chi^\perp} \lambda_i \rangle| \nonumber\\+ |\bra{\chi}\mu_i \rangle| + |\bra{\chi^\perp} \mu_i \rangle|],    \label{eq2}
\end{eqnarray}
over all pure states, $\ket{\chi}$, present in the Hilbert space, $\mathcal{H}_S$.
Here $\ket{\chi^\perp}$ denotes any pure state orthogonal to $\ket{\chi}$. Since the dimension of $\mathcal{H}_S$ is two, the set $\{\ket{\chi},\ket{\chi^\perp}\}$ will form a basis of $\mathcal{H}_S$. 
Let us denote the optimum basis for which $F(\rho_S,\Lambda_1,\Lambda_2)$ attains its maximum value as 
$\lbrace \ket{\chi_{\textnormal{opt}}}, \ket{\chi_{\textnormal{opt}}^{\perp}} \rbrace $. 
\item
Finally, we can construct two states as follows:
\begin{align}
\rho_{\textnormal{f}}^1(\rho_S,\Lambda_1,\Lambda_2) &\coloneqq  \lambda_1 \ket{\chi_{\textnormal{opt}}} +\lambda_2 \ket{\chi_{\textnormal{opt}}^\perp},\label{eq3}
\end{align}
and
\begin{align}
\rho_{\textnormal{f}}^2(\rho_S,\Lambda_1,\Lambda_2) &\coloneqq \mu_1 \ket{\chi_{\textnormal{opt}}} +\mu_2 \ket{\chi_{\textnormal{opt}}^\perp}.
\end{align}
\end{enumerate} 
Since the above states are found by constructing a basis that has maximum overlap with the eigenbases of both $\Lambda_2(t_2,t_1) \Lambda_1(t_1,0)\rho_S$ and $\Lambda_1(t_2,t_1) \Lambda_2(t_1,0)\rho_S$, we consider these states, $\rho_\textnormal{f}^{1}(\rho_S,\Lambda_1,\Lambda_2)$ and $\rho_\textnormal{f}^{2}(\rho_S,\Lambda_1,\Lambda_2)$, as the final states, obtained by acting $\mathcal{U}(\Lambda_1,\Lambda_2)$ on the state $\rho_S$.  

As one can notice, the action of $\mathcal{U}(\Lambda_1,\Lambda_2)$ on the state, $\rho_S$, has two definitions: one providing the state, $\rho^1_\textnormal{f}(\rho_S,\Lambda_1,\Lambda_2)$, and the other resulting in the state, $\rho^2_\textnormal{f}(\rho_S,\Lambda_1,\Lambda_2)$.
Since $\ket{\chi_{\textnormal{opt}}}$ and $\ket{\chi_{\textnormal{opt}}^\perp}$ may not be unique, the pair of states, $\rho_{\textnormal{f}}^1(\rho_S,\Lambda_1,\Lambda_2)$ and $\rho_\textnormal{f}^{2}(\rho_S,\Lambda_1,\Lambda_2)$, may also not be unique. 
Each and every one of these states will represent the action of $\mathcal{U}(\Lambda_1,\Lambda_2)$ on $\rho_S$.

We would like to note here that the states $\Lambda_2(t_2,t_1) \Lambda_1(t_1,0)\rho_S$ and $\Lambda_1(t_2,t_1) \Lambda_2(t_1,0)\rho_S$ can be determined without using the channel's Kraus decomposition. Hence, the construction of the UQS of any two or more quantum dynamics neither requires knowledge about the Kraus operators of any maps nor an external control system, like in the case of the CQS. Thus, we can construct a UQS for all quantum dynamics.

Here we have defined the UQS in such a way that it acts on two dynamics, $\Lambda_1$ and $\Lambda_2$, but the method can easily be generalized to form a UQS that can act on any arbitrary but fixed number of dynamics.

In the upcoming section, we will present an example where the UQS can be seen to be outperforming the CQS.

\section{Universal quantum switch provides an advantage in a state discrimination task}
\label{sdt}
In this section, we will discuss a state discrimination task in which, in certain parameter regions, the UQS will be witnessed to provide more advantages than the CQS. Let us consider the following scenario: Alice prepares a single-qubit state, $\tilde{\rho}$. She has a phase damping channel, $\Lambda_{PDC}(0.5)$, with noise strength $p_{PDC}=0.5$, and a unitary channel, $U$, which acts on two-dimensional systems. To avoid confusion, we would like to mention here that the number inside the braces in $\Lambda_{PDC}(0.5)$ represents the noise strength of the channel and not the time for which the evolution is taking place. In this scenario, Alice tosses a coin; if she gets the head, she applies both $\Lambda_{PDC}(0.5)$ and $U$; otherwise, she operates only $\Lambda_{PDC}(0.5)$ on the state $\tilde{\rho}$. She sends the final output state to Bob. Bob's task is to determine if the unitary has acted upon the state by performing measurements on the received state.
In these cases, when the result of the coin-toss is head, Alice can perform $\Lambda_{PDC}(0.5)$ and $U$ on $\tilde{\rho}$ either in a definite causal order or in an indefinite causal order. The operation of $\Lambda_{PDC}(0.5)$ and $U$ in indefinite causal order can again be implemented in two methods, viz., using the CQS or UQS. The exact difference between these three possible scenarios can be more clearly understood from the discussion below:
\begin{itemize}
    \item \textbf{\textit{Definite causal order:}} In this case, if the coin is head, Alice sends $\tilde{\rho}_1^{\text{DCO}}={U}(\Lambda_{PDC}(0.5) \tilde{\rho})U^\dagger$, and if it is tail, she sends $\tilde{\rho}_2=\Lambda_{PDC}(0.5) \tilde{\rho}$ to Bob.
    \item  \textbf{\textit{Conventional quantum switch:}} In this scenario, Alice transfers either the state, $\tilde{\rho}_1^{\text{CQS}}=\bra{\pm}\mathcal{S}(\Lambda_{PDC},{U}) (\tilde{\rho} \otimes \tilde{\omega}_C)\ket{\pm}$ or $\tilde{\rho}_2$ depending on the result of the coin toss, where $\tilde{\omega}_C=\ketbra{+}$ is the state of the control qubit. If the coin is head, she prepares and sends $\tilde{\rho}_1^{\text{CQS}}$, and otherwise, she sends $\tilde{\rho}_2$. See Eq.~\eqref{eq1} for more detail on the functional form of $\mathcal{S}$. 
    \item \textbf{\textit{Universal quantum switch:}} Here, if the coin is tail, Alice simply sends the state $\tilde{\rho}_2$ to Bob. Otherwise, Alice prepares the indefinite causal order of the two channels under consideration, i.e., $\Lambda_{PDC}(0.5)$ and $U$, following the model of the UQS, and acts the resulting map on $\tilde{\rho}$. Specifically, she prepares the state $\tilde{\rho}^{\text{UQS}}_1 =\rho_1^f(\tilde{\rho},U,\Lambda_{PDC})= \tilde{\lambda}_1 \ket{\tilde{\chi}_{\textnormal{opt}}} +\tilde{\lambda}_2 \ket{\tilde{\chi}_{\textnormal{opt}}^\perp}$, where $\{\ket{\tilde{\chi}_{\textnormal{opt}}},\ket{\tilde{\chi}_{\textnormal{opt}}^\perp}\}$ is the basis for which $F(\tilde{\rho},U,\Lambda_{PDC})$, defined in Eq.~\eqref{eq2}, reaches its maximum value. Here $\lbrace \tilde{\lambda}_i \rbrace_i$ is the set of eigenvalues of ${U}(\Lambda_{PDC}(0.5) \tilde{\rho})U^\dagger$.  
\end{itemize} 
Bob has knowledge about which state he would receive if Alice got a head in the coin toss and what state he would get in the case of a tail. Nevertheless, Bob does not have any information about the exact result of the coin toss. After receiving the state, Bob tries to find out if Alice got a head or tail, i.e., Bob, knowing $m$, tries to discriminate between the pair of states $\{\tilde{\rho}_1^m,\tilde{\rho}_2\}$, where $m=$DCO, CQS, or UQS.

Bob follows the minimum error state discrimination protocol to distinguish the states. The minimum probability of error in discriminating any two states using the minimum error state discrimination method has already been found, and it is generally known as the Helstrom bound~\cite{Helstrom,min-sd-2,min-sd-3,min-sd-4,min-sd-5,min-sd-6,min-sd-7}.
If two states, say $\rho_1$ and $\rho_2$, are prepared with the probability $p_1$ and $p_2=1-p_1$, respectively, the Helstrom bound tells us that the minimum probability of incorrectly guessing the state is $p_{\textnormal{err}} (\rho_1,\rho_2) \coloneqq \frac{1}{2} - \frac{1}{2} || p_1 \rho_1 - p_2 \rho_2||$.
We can utilize the Helstrom bound to determine the minimum probability of error in distinguishing the states $\{\tilde{\rho}_1^m,\tilde{\rho}_2\}$ for all $m$.
\begin{figure}[h!]
\includegraphics[scale=0.95]{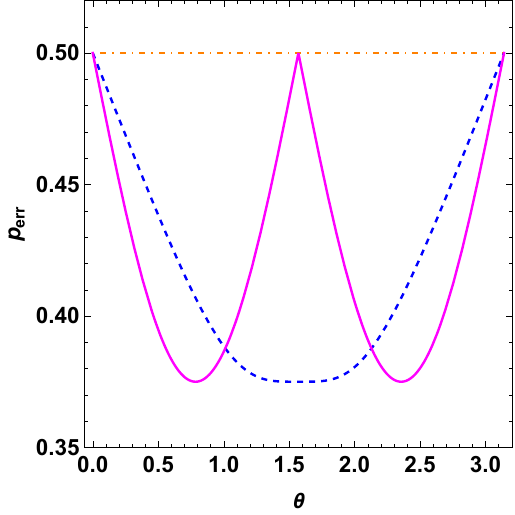}
\caption{\textbf{Comparative analysis between the definite causal order of dynamics and the indefinite causal order of channels modeled by conventional and universal quantum switches.} We present the parameter, $\theta$, involved in the unitary, $U$, and the minimum probability, $p_{\textnormal{err}}$, of error in discriminating the noisy state, $\tilde{\rho}_2$, and the unitarily transformed noisy state, $\tilde{\rho}_1^m$, on the horizontal and vertical axes, respectively, for $m=$DCO, CQS, and UQS. The dashed-dotted orange, dashed blue, and solid purple lines, respectively, show the behaviors of $p_{\textnormal{err}}\left(\tilde{\rho}_1^\text{DCO},\tilde{\rho}_2\right)$, $p_{\textnormal{err}}\left(\tilde{\rho}_1^\text{CQS},\tilde{\rho}_2\right)$, and $p_{\textnormal{err}}\left(\tilde{\rho}_1^\text{UQS},\tilde{\rho}_2\right)$ with respect to $\theta$. All quantities are dimensionless.}
\label{fig5} 
\end{figure}
Considering the coin of Alice to be unbiased, $\tilde{\rho}=\ket{+}\bra{+}$, and 
\begin{align}
U= \cos \theta \mathbbm{I}_2  - \iota \sin \theta \sigma_y,\nonumber
\end{align}
where $\theta \in \lbrack 0, \pi \rbrack$ is a constant, in Fig.~\ref{fig5}, we depict the behavior of the minimum probability of error in the state discrimination protocol against $\theta$.
Specifically, in the figure, we plot $p_{\textnormal{err}} \left(\tilde{\rho}_1^\text{DCO},\tilde{\rho}_2\right)$, $p_{\textnormal{err}} \left(\tilde{\rho}_1^\text{CQS},\tilde{\rho}_2\right)$ and $p_{\textnormal{err}} \left(\tilde{\rho}_1^\text{UQS},\tilde{\rho}_2\right)$ using dashed-dotted orange line, dashed blue line, and solid pink line, respectively. From the plot, it can be noticed that there exist two wide ranges of $\theta$, i.e., $\theta\in [0,1.01314)$ and $\theta\in(2.12846,\pi] $ for which $p_{\textnormal{err}} \left(\tilde{\rho}_1^\text{UQS},\tilde{\rho}_2\right)$ is the smallest among the three error probabilities. Therefore, within these ranges, the application of $U$ and $\Lambda_{PDC}(0.5)$ in indefinite causal order defined using the method of the UQS provides the most advantage compared to the other cases. Thus, we realize that in certain circumstances, the UQS can exhibit a benefit over the CQS.

\section{universal quantum switch can retain CP-indivisibility when acting upon CP-indivisible maps}
In this section, we will form an indefinite causal order of two CP-indivisible dynamics using a quantum switch and show that the resulting dynamics remains CP-indivisible.
\label{sec4}
In this regard, let us take a composite set-up consisting of a system, $S$, and an environment, $E$. We consider both of them to have dimension two. Let the initial state of the environment be $\rho_E$, which is in product with the system, $S$. At a moment, the system begins to interact with the environment. Hence, the entire set-up, $SE$, starts to unitarily evolve with time. Let there be two unitaries, $U_1(t_2,0)=\exp(-\iota H_1t_2/\hbar)$ and $U_2(t_2,0)=\exp(-\iota H_2t_2/\hbar)$, which dictate the evolution of $SE$ for the time interval $t_2$. Here $\hbar$ represents Planck's constant. The two Hamiltonians, $H_1$ and $H_2$, are given by
\begin{align*}
H_1 \coloneqq h\sigma_z \otimes \mathbbm{I}_2 + h\mathbbm{I}_2 \otimes \sigma_z + J_1 \sigma_x \otimes \sigma_x,\\
H_2 \coloneqq h\sigma_z \otimes \mathbbm{I}_2 + h\mathbbm{I}_2 \otimes \sigma_z + J_2 \sigma_x \otimes \sigma_x,
\end{align*}
where $J_1 = 0.5h$ and $J_2 = h$.
Since the considered unitaries are entangling, the effective evolution of the system, $S$, will certainly not be unitary.  Let us denote the system's evolution due to the interaction strengths, $J_1$ and $J_2$, within time $[0,t_2]$, as $\Lambda_{H_1}(t_2,0)$ and $\Lambda_{H_2}(t_2,0)$.
Considering the two following initial states of the system
\begin{align*}
\sigma_S^1 &= \left( \begin{array}{cc}
\frac{1}{2}(1+\frac{1}{\sqrt{2}}) & \frac{1}{2 \sqrt{2}} \\
\frac{1}{2 \sqrt{2}} &
\frac{1}{2}(1-\frac{1}{\sqrt{2}}) \end{array} \right),\quad \sigma_S^2 = \frac{1}{2} \left( \begin{array}{cc}
1 & 1 \\
1 & 1 \end{array} \right),
\end{align*}
and fixing the initial state of the environment at $\rho_E=\ket{0}\bra{0}$, we can calculate $D(\Lambda_{H_i}(t_2,0)\sigma^1_S, \Lambda_{H_i}(t_2,0)\sigma^2_S)$. From $D(\Lambda_{H_i}(t_2,0)\sigma^1_S, \Lambda_{H_i}(t_2,0)\sigma^2_S)$ it can be easily checked using the P-divisibility criteria, discussed in Sec. \ref{sec2}, that $\Lambda_{H_i}$ is P-indivisible for both $i=1$ and 2. Since P-indivisibility implies CP-indivisibility, we can conclude both $\Lambda_{H_1}$ and $\Lambda_{H_2}$ are CP-indivisible as well.

We want to find the indefinite causal order of the two dynamics, $\Lambda_{H_1}$ and $\Lambda_{H_2}$, on the system, $S$. In this regard, we divide the entire time range, $[0,t_2]$, into two parts, viz., $[0,t_1)$ and $[t_1,t_2]$. Within the two time intervals, the system interacts with the same environment with different strengths, either $J_1$ or $J_2$. We will refer to the maps acting on the system between time $t_1$ and $t_2$ because of the system's interaction with $E$ of strengths $J_1$ and $J_2$, as $\Lambda_{H_1}(t_2,t_1)$ and $\Lambda_{H_2}(t_2,t_1)$, respectively.
One needs to keep in mind that, since at time $t_1$, $S$ may be entangled with $E$, the maps, $\Lambda_{H_1}(t_2,t_1)$ and $\Lambda_{H_2}(t_2,t_1)$, may not be CPTP and therefore will not have any Kraus operator decomposition. Therefore, we cannot use CQS to generate the indefinite causal order of dynamics. Hence, we will use the UQS to form the indefinite causal orders of the two dynamics, $\Lambda_{H_1}$ and $\Lambda_{H_2}$. 
\begin{figure}[h!]
\includegraphics[scale=0.95]{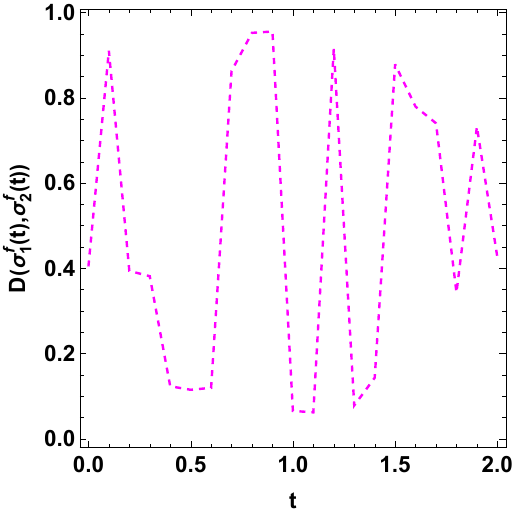}
\caption{\textbf{Preservation of CP-indivisibility under the application of the universal quantum switch}.
We plot trace distance, $D(\sigma_f^1, \sigma_f^2)$, between the two states, $\sigma_f^1$ and $\sigma_f^2$, along the vertical axis with time, $t$, shown along the horizontal axis, using a pink dashed line. The visible non-monotonicity of the pink dashed curve with $t$ proves that the dynamics generated by operating UQS on two CP-indivisible dynamics, $\Lambda_{H_1}$ and $\Lambda_{H_2}$, is CP-indivisible. The horizontal axis is in the units of $\hbar/h$, and the vertical axis is dimensionless.}
\label{fig6} 
\end{figure}
To examine if $\mathcal{U}(\Lambda_{H_1},\Lambda_{H_2})$ is CP-indivisible, we consider the same initial states, $\sigma_S^1$ and $\sigma_S^2$, of the system, $S$, and environment, $\rho_E$. We first evaluate the states: $\Lambda_{H_2}(t_2,t_1)\Lambda_{H_1}(t_1,0)\sigma_S^i=\tr_E\left[{U_2(t_2,t_1)U_1(t_1,0)\sigma_S^i\otimes \rho_E U^\dagger_1(t_1,0) U^\dagger_2(t_2,t_1)}\right]$ and $\Lambda_{H_1}(t_2,t_1)\Lambda_{H_2}(t_1,0)\sigma_S^i=\tr_E\left[{U_1(t_2,t_1)U_2(t_1,0)\sigma_S^i\otimes\rho_E U^\dagger_2(t_1,0) U^\dagger_1(t_2,t_1)}\right]$ for $i=1$ and 2, and find their spectral decompositions. Finally, using their eigen decomposition, we determine $\rho^1_\textnormal{f}(\sigma_S^1,\Lambda_{H_1},\Lambda_{H_2})$ and $\rho^1_\textnormal{f}(\sigma_S^2,\Lambda_{H_1},\Lambda_{H_2})$ and denote them as $\sigma^1_\textnormal{f}$ and $\sigma^2_\textnormal{f}$, respectively. Hence, $\sigma^1_\textnormal{f}$ and $\sigma^2_\textnormal{f}$ are the final states after the action of $\mathcal{U}(\Lambda_{H_1},\Lambda_{H_2})$ on $\sigma^1_S$ and $\sigma^2_S$, respectively.
We would like to mention here that the explicit forms of the states, $\sigma^1_\textnormal{f}$ and $\sigma^2_\textnormal{f}$, depend on the time parameters, $t_1$ and $t_2$. Considering $t_1=t$ and $t_2=2t$, in Fig.~\ref{fig6}, we plot $D(\sigma^1_\textnormal{f}, \sigma^2_\textnormal{f})$, as a function of $t$. It is noticeable from the figure that there are various instances when $D(\sigma^1_\textnormal{f}, \sigma^2_\textnormal{f})$ increases with $t$, indicating the non-monotonic behavior of $D(\sigma^1_\textnormal{f}, \sigma^2_\textnormal{f})$. Hence, it can be concluded that the dynamics, $\mathcal{U}(\Lambda_{H_1},\Lambda_{H_2})$, is CP-indivisible.

\section{When do quantum switches preserve CP-divisibility?}
\label{sec5}
In this section, we investigate the necessary and sufficient condition for the indefinite causal orders of two arbitrary CP-divisible dynamics, $\Lambda_1^{CP}$ and $\Lambda_2^{CP}$, to be CP-divisible.
We would like to note here that since $\Lambda_1^{CP}$ ($\Lambda_2^{CP}$) is CP-divisible, $\Lambda_1^{CP}(t_2,t_1)$ ($\Lambda_2^{CP}(t_2,t_1)$) will always have Kraus operator decomposition for all $t_2\geq t_1$. Hence to prepare the indefinite causal order we can utilize the CQS, $\mathcal{S}$. We focus on the condition of the
dynamics, $\{\Sigma(\Lambda_1^{CP},\Lambda_2^{CP},t_2,0,\cdot )\}_{t_2}$, created by acting $\mathcal{S}$ on the CP-divisible dynamics, $\Lambda_1^{CP}$ and $\Lambda_2^{CP}$,  
to be CP-divisible. 
We also provide two examples of pairs of CP-divisible dynamics first of which remains CP-divisible when a CQS acts on them and the other becomes CP-indivisible when inserted in the CQS. Furthermore, we discuss examples of pairs of CP-divisible dynamics on which, when a CQS acts, the resulting dynamics not only becomes CP-indivisible but also P-indivisible.

We know, the dynamics, $\{\Sigma(\Lambda_1^{CP},\Lambda_2^{CP},t_2,0,\cdot )\}_{t_2}$, constructed by applying the CQS on two CP-divisible dynamics, $\Lambda_1^{CP}$ and $\Lambda_2^{CP}$, will be referred to as CP-divisible, if and only if 
\begin{eqnarray}
    &{}_C \bra{\pm} S(\phi_1, \phi_2)(t_2,0) (\rho_S \otimes \omega_C) \ket{\pm}_C\\ &= {}_C \bra{\pm} S(\phi_1, \phi_2) (t_2, t_1) S(\phi_1, \phi_2) (t_1,0) (\rho_S \otimes \omega_C) \ket{\pm}_C\nonumber \label{extra2}
\end{eqnarray}
 holds for each and every time $t_2 \geq t_1 \geq 0$ and all quantum states, $\rho_S$, with a control qubit, $\omega_C$, where ${}_C \bra{\pm} S(\phi_1, \phi_2) (t_2, t_1) (\rho_S \otimes \omega_C) \ket{\pm}_C$ is a complete positive map.} 

In the next part, we state and prove the necessary and sufficient condition for the dynamics, $\{\Sigma(\Lambda_1^{CP},\Lambda_2^{CP},t_2,0,\cdot )\}_{t_2}$, built by acting the CQS on two arbitrary CP-divisible dynamics, $\Lambda_1^{CP}$ and $\Lambda_2^{CP}$, to be CP-divisible. 
\begin{lemma}
\label{proof-final-CP-divisible}
    The dynamics, $\{\Sigma(\Lambda_1^{CP},\Lambda_2^{CP},t_2,0,\cdot )\}_{t_2}$, constructed by applying the CQS on two CP-divisible dynamics, $\Lambda_1^{CP}$ and $\Lambda_2^{CP}$, is CP-divisible, if and only if, the Kraus operators, $\{K_i^{(l)} (s,t_1)\}_i$ and $\{K_i^{(l)} (t_2,s)\}_i$, of $\Lambda_l^{CP}(s,t_1)$ and $\Lambda_l^{CP}(t_2,s)$ that describe the evolution of the system, respectively, in between the time, $t_1$, $s$ and $s$, $t_2$, satisfy 
    \begin{align}
    \label{commute-cp-d-eqn}
        K_i^{(1)} (t_2,s) K_j^{(2)} (s,t_1) = K_j^{(2)} (t_2,s) K_i^{(1)} (s,t_1),
    \end{align}
 for all $s$, $t_1$, $t_2$, $i$, $j$ that obey $t_2 \geq s \geq t_1$. Here $l$ takes two values, 0 and 1, representing the two different dynamics, $\Lambda_1^{CP}$ and $\Lambda_2^{CP}$. 
\end{lemma}

\begin{proof}
The exact expressions of $\Sigma(\Lambda_1^{CP},\Lambda_2^{CP},t_2,0,\rho_S )$ and $\Sigma(\Lambda_1^{CP},\Lambda_2^{CP},t_2,t_1,\Sigma(\Lambda_1^{CP},\Lambda_2^{CP},t_1,0,\rho_S ) )$ are given in the Appendix, where $\rho_S$ is the initial state of the system under consideration. 
If Eq.~\eqref{commute-cp-d-eqn} be fulfilled by all the Kraus operators of the maps, $\Lambda_l^{CP}(s,t_1)$ and $\Lambda_l^{CP}(t_2,s)$, for all $t_1$, $t_2$, $s$, and $l=1$ and 2, it is evident that the requirements, mentioned in Eq.~\eqref{extra2}, for the action of the CQS on two CP-divisible dynamics to be CP-divisible will be satisfied. Hence, it proves that Eq.~\eqref{commute-cp-d-eqn} acts as a sufficient condition for the action of the CQS on two CP-divisible dynamics to provide a CP-divisible dynamics.

Let us now move to the necessary condition. In this regard, we consider the dynamics, $\{\Sigma(\Lambda_1^{CP},\Lambda_2^{CP},t_2,0,\cdot )\}_{t_2}$, to be CP-divisible. Hence, $\Sigma(\Lambda_1^{CP},\Lambda_2^{CP},t_2,t_1,\cdot )$, must be CPTP for all intermediate times, $t_1$, which satisfy $t_2 \geq t_1 \geq 0$. Therefore, we can write
\begin{align*}
\Sigma(\Lambda_1^{CP},\Lambda_2^{CP},t_2,t_1,\rho_S ) = \sum_q M_q(t_2,t_1) \rho_S \left[M_q(t_2,t_1)\right]^\dagger,
\end{align*}
for all valid states, $\rho_S$, of the system, $S$, where $\{M_q(t_2,t_1)\}_q$ is a set of Kraus operators. Substituting the exact expression of $\Sigma(\Lambda_1^{CP},\Lambda_2^{CP},t_2,t_1,\rho_S )$ (written in the Appendix) in the left-hand side of the above equation and taking trace on both sides, we have
\begin{widetext}
\begin{eqnarray*}
 &&\frac{1}{2} \Bigg[ \sum_{ij} \Tr \Bigg[ \left[K_j^{(2)} (s,t_1)\right]^\dagger \left[K_i^{(1)} (t_2,s)\right]^\dagger K_j^{(2)} (t_2,s) K_i^{(1)}(s,t_1) \rho_S   \nonumber \\ && + \left[K_i^{(1)}(s,t_1)\right]^\dagger \left[K_j^{(2)}(t_2,s)\right]^\dagger K_i^{(1)}(t_2,s)  K_j^{(2)}(s,t_1) \rho_S \Bigg] \Bigg]
  = \Tr \left[ \rho_S \right].
\end{eqnarray*} 
\end{widetext}
Here we have used the fact that $\Lambda_1^{CP}$ and $\Lambda_2^{CP}$ are two CP-divisible dynamics.
Since the above equation is true for all $\rho_S$, we can conclude that 
\begin{widetext}
\begin{align*}
& \frac{1}{2} \left[ \sum_{ij} \Tr \left[ \left[K_j^{(2)} (s,t_1)\right]^\dagger \left[K_i^{(1)} (t_2,s)\right]^\dagger K_j^{(2)} (t_2,s) K_i^{(1)}(s,t_1)   + \left[K_i^{(1)}(s,t_1)\right]^\dagger \left[K_j^{(2)}(t_2,s)\right]^\dagger K_i^{(1)}(t_2,s)  K_j^{(2)}(s,t_1) \right] \right]= \mathbb{I}, 
\end{align*} 
\end{widetext}
where $\mathbbm{I}$ represents the identity operator which acts on the system's Hilbert space. By simplifying the above expression we get
\begin{align*}
 \sum_{ij} Y_{ij}^\dagger Y_{ij} = 0,\label{commute-cp-d-proof-5}
\end{align*}
where $Y_{ij} = K_j^{(2)}(t_2,s) K_i^{(1)}(s,t_1) - K_i^{(1)}(t_2,s) K_j^{(2)}(s,t_1)$. From the above equation, we can write $\sum_{ijk} |\lambda_{ij}^k|^2 =0$, where $\lambda_{ij}^k$ is the $k$th eigenvalue of $Y_{ij}$. Hence we get $\lambda_{ij}^k=0$ for all $i$, $j$, $k$, which implies $Y_{ij}=0$ for all $i$ and $j$. Therefore we have 
\begin{align*}
    K_i^{(1)} (t_2,s) K_j^{(2)} (s,t_1) = K_j^{(2)} (t_2,s) K_i^{(1)} (s,t_1).
\end{align*}
Thus we get the condition expressed in Eq.~\eqref{commute-cp-d-eqn},  is necessary to hold for the dynamics, $\{\Sigma(\Lambda_1^{CP},\Lambda_2^{CP},t_2,0\cdot )\}_{t_2}$, to be CP-divisible. 
Hence it completes the proof.
\end{proof}

\textbf{Remark 1.} Despite the fact that the Kraus operator decomposition of any CPTP map is not unique~\cite{chuang}, the action of the CQS on two CPTP dynamics remains unique, whatever the considered Kraus decomposition of the maps that are participating in the switch~\cite{switch-first-paper}. Since different Kraus operator decompositions of CPTP maps are unitarily connected with each other~\cite{chuang}, if Eq.~\eqref{commute-cp-d-eqn} holds for particular Kraus operator decompositions of a pair of maps, it will be satisfied for all Kraus operator decompositions of those pair of maps. Therefore, the stated necessary and sufficient condition does not depend on which set of Kraus operators is being considered.\vspace{2.0mm}

\textbf{Remark 2.} The theorem can be generalized for any CQS that acts on an arbitrary but fixed number of CP-divisible dynamics. If we act CQS on $N$ CP-divisible dynamics, the resulting dynamics will form a superposition of $N!$ causal orders of the input dynamics. Hence, the satisfaction of the commutativity relation, as expressed in Eq.~\eqref{commute-cp-d-eqn}, among Kraus operators of each and every pair of the set of $N$ maps, for all times, jointly serves as a necessary and sufficient criterion for maintaining CP-divisibility of the dynamics created by applying the CQS on $N$ CP-divisible dynamics.\vspace{2.0mm}

Let us now come to some examples. The first example is of a dynamics generated by acting the CQS on two phase damping channels. Since the Kraus operators, $K_1^{(1)}(t_2,t_1)=e^{-\Gamma_{PDC}^{(1)} (t_2-t_1)/2}\mathbb{I}_2$, $K_2^{(1)}(t_2,t_1)=\sqrt{1-e^{-\Gamma_{PDC}^{(1)} (t_2-t_1)}}\sigma_z$,  and $K_1^{(2)}(t_2,t_1)=e^{-\Gamma_{PDC}^{(2)} (t_2-t_1)/2}\mathbb{I}_2$, $K_2^{(2)}(t_2,t_1)=\sqrt{1-e^{-\Gamma_{PDC}^{(2)} (t_2-t_1)}}\sigma_z$ of two ideal phase damping channels satisfy Eq.~\eqref{commute-cp-d-eqn}, where $\Gamma_{PDC}^{(1)}$ and $\Gamma_{PDC}^{(2)}$ are the Lindblad coefficients of the channels, we arrive at the following conclusion:
\begin{example}
The dynamics produced by applying the CQS on two phase damping channels is CP-divisible.
\end{example}

One can notice from the expressions presented in Eqs.~\eqref{extra1} and~\eqref{adc} that the Kraus operators of the amplitude damping channel do not commute with each other. Hence, they do not satisfy Eq.~\eqref{commute-cp-d-eqn}. Therefore, we get
\begin{example}
The dynamics formed by applying the CQS on two ideal amplitude damping channels is CP-indivisible.
\end{example}

The examples assure that the fulfillment of the necessary and sufficient condition given in Eq.~\eqref{commute-cp-d-eqn} can be easily verified for quantum dynamics. 

Let us now try to check if it is possible to prepare the P-indivisible dynamics by operating the CQS on pairs of CP-divisible dynamics. We will make use of Eq.~\eqref{eq-BLP} to detect the P-indivisibility of the dynamics constructed by the CQS.
In all the examples discussed below, to apply the CQS using the definition given in Eq.~\ref{eq1}, we take $t_1=t$ and $t_2=2t$, where $t$ is the half of the total evolution time of the system.

We know an ideal depolarizing channel is CP-divisible. Regardless of whether we consider the Lindblad coefficients of a pair of depolarizing channels to be the same or different, the Kraus operators of the channels are not going to commute. Therefore, according to Theorem \ref{proof-final-CP-divisible}, the operation of CQS on any two depolarizing channels will produce CP-indivisible dynamics. Let us further check if the produced dynamics is P-indivisible. 
In this regard, we take two initial states, $\rho_S^1 = \ket{0} \bra{0}$ and $\rho_S^2 = \ket{1}\bra{1}$, of a system on which the depolarizing channels can act. First, we consider the two depolarizing channels to have equal Lindblad coefficients, $\Gamma_{DC}= \Gamma$. 
\begin{figure}[h!]
\includegraphics[scale=0.95]{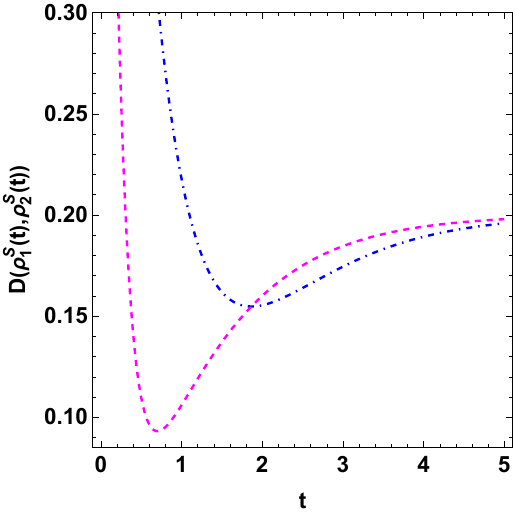}
\caption{\textbf{Exploring the behavior of  dynamics produced by acting the conventional quantum switch on two depolarizing channels}.
In the vertical axis, we plot $D(\rho_S^1(t),\rho_S^2(t))$ with respect to $t$, presented along the horizontal axis, where $\rho_S^1(t)=\Sigma(\Lambda_{DC},\Lambda_{DC},t,0,\rho_S^1)$ and $\rho_S^2(t)=\Sigma(\Lambda_{DC},\Lambda_{DC},t,0,\rho_S^2)$. The blue dash-dotted and pink dashed lines represent the scenarios where the two depolarizing channels on which the CQS is being applied have the same Lindblad coefficients, $\Gamma_{DC}=\Gamma$, and different Lindblad coefficients, i.e., $\Gamma_{DC}=\Gamma$ and $\Gamma_{DC}=5\Gamma$, respectively. The vertical axis is dimensionless, whereas the horizontal axis has units of $\Gamma^{-1}$.}
\label{fig2} 
\end{figure}
We evaluate $\rho_S^1(t)=\Sigma(\Lambda_{DC},\Lambda_{DC},t,0,\rho_S^1)$ and $\rho_S^2(t)=\Sigma(\Lambda_{DC},\Lambda_{DC},t,0,\rho_S^2)$ and plot the trace distance, $D(\rho_S^1(t),\rho_S^2(t))$, between $\rho_S^1 (t)$ and $\rho_S^2 (t)$ with time, $t$, in Fig. \ref{fig2}, using blue dash-dotted curve.
One can notice from the curve that $D(\rho_S^1(t),\rho_S^2(t))$ shows non-monotonic behavior with respect to $t$. This characteristic proves that the operation of the CQS on two depolarizing channels with the same Lindblad coefficients can produce P-indivisible dynamics.
Next, we apply the CQS on two depolarizing channels having unequal Lindblad coefficients, i.e., one with $\Gamma_{DC} = \Gamma$ and the other with $\Gamma_{DC} = 5\Gamma$.
For these two considered depolarizing channels, we again plot the trace distance, $D(\rho_S^1(t),\rho_S^2(t))$, in Fig. \ref{fig2}, using a pink dashed line, where $\rho_S^1(t)=\Sigma(\Lambda_{DC},\Lambda_{DC},t,0,\rho_S^1)$ and $\rho_S^2(t)=\Sigma(\Lambda_{DC},\Lambda_{DC},t,0,\rho_S^2)$. In this scenario also, we can notice a non-monotonic behavior in the pink dashed curve with respect to time, $t$, proving the P-indivisible nature of the dynamics, $\{\Sigma(\Lambda_{DC},\Lambda_{DC},t,0,\cdot)\}_t$. Hence, we get to the following conclusion:
\begin{example}
The action of the CQS on two CP-divisible dynamics, which are ideal depolarizing channels having the same or different Lindblad coefficients, can produce a P-indivisible dynamics.
\end{example}

\begin{figure}[h!]
\includegraphics[scale=0.95]{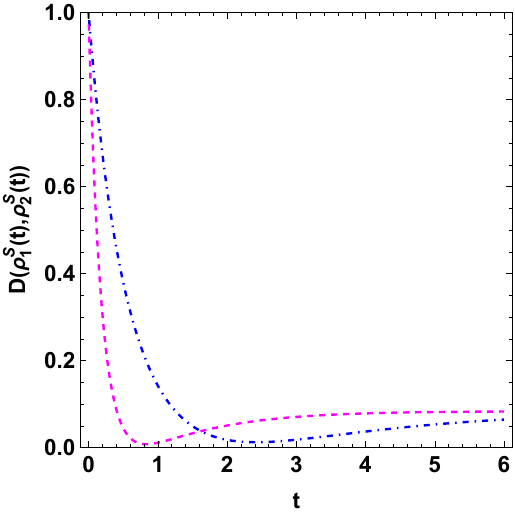}
\caption{\textbf{Nature of the P-indivisibility of dynamics generated by applying the conventional quantum switch on two CP-divisible dynamics, viz. ideal depolarizing and amplitude damping channels}. The quantity, $D(\rho_S^1(t),\rho_S^2(t))$, is plotted along the vertical axis against $t$, represented along the horizontal axis. Here $\rho_S^1(t)$ ($\rho_S^2(t)$) is the state found by acting the map $\Sigma(\Lambda_{DC},\Lambda_{ADC},t,0,\cdot)$ on $\rho_S^1$ ($\rho_S^2$). The blue dash-dotted and pink dashed lines denote the cases where the considered channels, i.e., the ideal depolarizing and amplitude damping noise, have, respectively, equal Lindblad coefficients ($\Gamma_{DC}=\Gamma_{ADC}=\Gamma$) and unequal Lindblad coefficients ($\Gamma_{DC}=\Gamma$ and $\Gamma_{ADC}=5\Gamma$). The quantity depicted in the horizontal axis has units of $\Gamma^{-1}$. The vertical axis is dimensionless.}
\label{fig3} 
\end{figure}

We also apply the CQS on two pairs of channels, viz., ideal depolarizing and amplitude damping channels having the same Lindblad coefficients, $\Gamma_{DC} = \Gamma_{ADC} = \Gamma$, and the same pair of channels having unequal Lindblad coefficients, i.e., $\Gamma_{DC} = \Gamma$ and $\Gamma_{ADC} = 5\Gamma$. Since the Kraus operators of the depolarizing maps do not commute with the same of amplitude damping channels in both situations, the effective dynamics created by the CQS
will be CP-indivisible. We can verify the P-indivisibility of the dynamics, $\{\Sigma(\Lambda_{DC},\Lambda_{ADC},t,0,\cdot)\}_t$ in the same way as in the previous example.
In this regard, we consider the same initial states, $\rho_S^1=\ketbra{0}$ and $\rho_S^2=\ketbra{1}$. In Fig. \ref{fig3}, we illustrate the trace distances, $D(\rho_S^1(t),\rho_S^2(t))$, between the states $\rho_S^1(t)=\Sigma(\Lambda_{DC},\Lambda_{ADC},t,0,\rho_S^1)$ and $\rho_S^2(t)=\Sigma(\Lambda_{DC},\Lambda_{ADC},t,0,\rho_S^2)$ for the cases where the Lindblad coefficients have equal values (dash-dotted blue line) and unequal values (dashed pink line) with respect to $t$. The clear non-monotonic nature of the curves over $t$ for both cases confirms the following result:
\begin{example}
The dynamics produced by applying the CQS on two CP-divisible dynamics, viz., ideal amplitude damping and depolarizing channels having the same or different Lindblad coefficients, can be P-indivisible.
\end{example}

\begin{figure}[h!]
\includegraphics[scale=0.95]{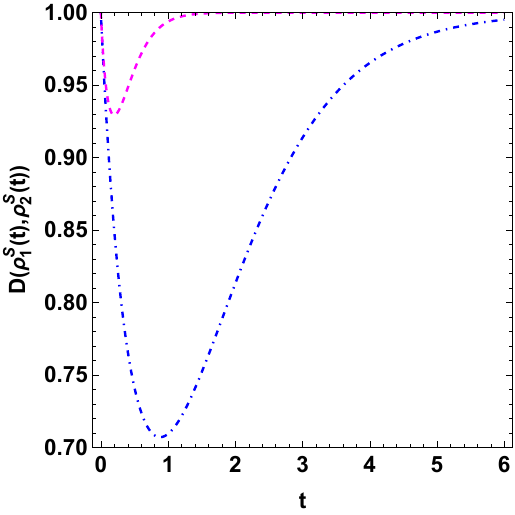}
\caption{\textbf{Illustration of the P-indivisibility characteristic of the dynamics built by applying conventional quantum switch on ideal depolarizing and phase damping channels.} All the considerations are the same as in Fig. \ref{fig3} except for the fact that here the CQS acts on ideal depolarizing and phase damping channels, and the Lindblad coefficients of the channels are $\Gamma_{DC}=\Gamma_{PDC}=1$ (shown using a blue dash-dotted line) and $\Gamma_{DC}=1$ and $\Gamma_{PDC}=5$ (depicted using a pink dashed curve).}
\label{fig4} 
\end{figure}

As the final example, in Fig.~\ref{fig4}, we plot the trace distance, $D(\rho_S^1(t),\rho_S^2(t))$, between $\rho_S^1(t)=\Sigma(\Lambda_{DC},\Lambda_{PDC},t,0,\rho_S^1)$ and $\rho_S^2(t)=\Sigma(\Lambda_{DC},\Lambda_{PDC},t,0,\rho_S^2)$, where $\rho_S^1$ and $\rho_S^2$ are the same initial states as considered in the previous two examples. The blue dash-dotted and pink dashed curves depict scenarios where the two channels, $\Lambda_{DC}$ and $\Lambda_{PDC}$, have equal Lindblad coefficients, $\Gamma_{DC}=\Gamma_{PDC}=\Gamma$ and unequal Lindblad coefficients, i.e., $\Gamma_{DC}=\Gamma$ and $\Gamma_{PDC}=5\Gamma$, respectively. From the non-monotonic behaviors of $D(\rho_S^1(t),\rho_S^2(t))$ with respect to $t$, which can be clearly witnessed in the two curves of Fig. \ref{fig4}, we make the following statement:
\begin{example}
The operation of the CQS on two CP-divisible dynamics, the ideal depolarizing and phase-damping channels having the same or different Lindblad coefficients, can create a P-indivisible dynamics.
\end{example}

\section{conclusion}
\label{sec6}
A quantum switch is a superoperator that creates a superposition of different causal orders of quantum channels using a control system. In quantum informational tasks,  dynamics generated by applying a quantum switch on two channels have been proven to provide  advantage over the utilization of the channels in definite causal order.

{Due to the inability to decompose the intermediate evolution of a system undergoing CP-indivisible dynamics into Kraus operators, the operation of a traditional quantum switch on CP-indivisible dynamics is infeasible.}
We provided the notion of a universal quantum switch that can be applied 
on all types of quantum dynamics. Our approach allowed for the construction of the indefinite causal order of any set of quantum dynamics without limiting it to CP-divisible dynamics.
We further presented a state discrimination task where the UQS was proven to perform better than the usual quantum switch in certain parameter regions. Using the model of the UQS, we showed an example where the indefinite causal order of two CP-indivisible channels is CP-indivisible. Moreover, we investigated whether the typical quantum switch can preserve CP-divisibility when superposing the different causal orders of CP-divisible dynamics. We analytically proved that the commutativity of the Kraus operators of one of the CP-divisible channels with the Kraus operators of the other CP-divisible channel for all time is a necessary and sufficient criterion for the channel generated by the action of the CQS on these channels to be CP-divisible. Finally, we found a few examples of CP-divisible dynamics transforming into P-indivisible dynamics when {acted upon} by the typical quantum switch.

\section*{Acknowledgment}
KS acknowledges support from the project MadQ-CM (Madrid Quantum de la Comunidad de Madrid) funded by the European Union (NextGenerationEU, PRTR-C17.I1) and by the Comunidad de Madrid (Programa de Acciones Complementarias).

\section*{Appendix}

\begin{appendix}
Let us assume that $\lbrace K_j^{(1)}(t_{f},t_{i}) \rbrace _j$ and $\lbrace K_j^{(2)}(t_{f},t_{i})  \rbrace_j$ are the Kraus operators of the two maps, $\Lambda_1^{CP}(t_{f},t_{i})$ and $\Lambda_1^{CP}(t_{f},t_{i})$, that evolves the system from any initial time, $t_{i}$, to a final time, $t_{f}$. Here we have considered the state of the control as $\omega_C = \ket{+}$.

\begin{widetext}
\begin{align*}
 &\Sigma(\Lambda_1^{CP}, \Lambda_2^{CP},t_2,t_1,\rho_S)\\  &= \frac{1}{4} \sum_{ij} \Bigg[ K_j^{(2)} (t_2,s) K_i^{(1)} (s,t_1) \rho_S \left[K_i^{(1)}(s,t_1)\right]^\dagger \left[K_j^{(2)} (t_2,s)\right]^\dagger + K_j^{(2)} (t_2,s) K_i^{(1)}(s,t_1) \rho_S \left[K_j^{(2)} (s,t_1)\right]^\dagger \left[K_i^{(1)} (t_2,s)\right]^\dagger \nonumber \\ &+ K_i^{(1)} (t_2,s) K_j^{(2)} (s,t_1) \rho_S \left[K_j^{(2)}(s,t_1)\right]^\dagger \left[K_i^{(1)} (t_2,s)\right]^\dagger + K_i^{(1)}(t_2,s) K_j^{(2)}(s,t_1) \rho_S \left[K_i^{(1)}(s,t_1)\right]^\dagger \left[K_j^{(2)}(t_2,s)\right]^\dagger\Bigg],\\ \nonumber
 \end{align*}
where $t_2 \geq s \geq t_1 \geq 0$. Then, we can write $\Sigma(\Lambda_1^{CP}, \Lambda_2^{CP},t_2,t_1,\Sigma(\Lambda_1^{CP}, \Lambda_2^{CP},t_1,0,\rho_S))$ as follows:

 \begin{align*}
 &\Sigma(\Lambda_1^{CP}, \Lambda_2^{CP},t_2,t_1,\Sigma(\Lambda_1^{CP}, \Lambda_2^{CP},t_1,0,\rho_S)) \\ &= \frac{1}{4} \sum_{ijmn} \Bigg[ K_n^{(2)}(t_2,s_2) K_m^{(1)}(s_2,t_1) K_j^{(2)}(t_1,s_1) K_i^{(1)}(s_1,0) \rho_S \left[K_i^{(1)}(s_1,0)\right]^\dagger \left[K_j^{(2)}(t_1,s_1)\right]^\dagger \left[K_m^{(1)}(s_2,t_1)\right]^\dagger \left[K_n^{(2)}(t_2,s_2)\right]^\dagger \nonumber \\&+
 K_n^{(2)}(t_2,s_2) K_m^{(1)}(s_2,t_1) K_j^{(2)}(t_1,s_1) K_i^{(1)}(s_1,0) \rho_S \left[K_j^{(2)}(s_1,0)\right]^\dagger \left[K_i^{(1)}(t_1,s_1)\right]^\dagger \left[K_n^{(2)}(s_2,t_1)\right]^\dagger \left[K_m^{(1)}(t_2,s_2)\right]^\dagger \nonumber \\ & + 
 K_m^{(1)}(t_2,s_2) K_n^{(2)}(s_2,t_1) K_i^{(1)}(t_1,s_1) K_j^{(2)}(s_1,0) \rho_S \left[K_j^{(2)}(s_1,0)\right]^\dagger \left[K_i^{(1)}(t_1,s_1)\right]^\dagger \left[K_n^{(2)}(s_2,t_1)\right]^\dagger \left[K_m^{(1)}(t_2,s_2)\right]^\dagger \nonumber \\ &+
 K_m^{(1)}(t_2,s_2) K_n^{(2)}(s_2,t_1) K_i^{(1)}(t_1,s_1) K_j^{(2)}(s_1,0) \rho_S \left[K_i^{(1)}(s_1,0)\right]^\dagger \left[K_j^{(2)}(t_1,s_1)\right]^\dagger \left[K_m^{(1)}(s_2,t_1)\right]^\dagger \left[K_n^{(2)}(t_2,s_2)\right]^\dagger \Bigg],
\end{align*}
\end{widetext}
for $t_2 \geq s_2 \geq t_1 \geq s_1 \geq 0$.

\end{appendix}

\bibliography{switch_qm}
\end{document}